\documentclass[12pt,reqno]{amsart}
\usepackage{epstopdf}
\usepackage[pagewise]{lineno}
\usepackage[caption=false]{subfig}

\usepackage[numbers,sort&compress]{natbib}
\bibpunct[, ]{[}{]}{,}{n}{,}{,}

   \newcommand{\scal}[1]{{\left\langle{#1}\right\rangle}}
\newcommand{\Z}{\mathbb Z} 
  \newcommand{\bw}{\overline{w}}
\theoremstyle{plain}
\newtheorem{theorem}{Theorem}[section]
\newtheorem{lemma}[theorem]{Lemma}
\newtheorem{corollary}[theorem]{Corollary}
\newtheorem{proposition}[theorem]{Proposition}

\theoremstyle{definition}

\theoremstyle{remark}
\newtheorem{remark}[theorem]{Remark}

\usepackage[colorlinks=true,linkcolor=red]{hyperref}%
\usepackage[usenames,dvipsnames]{color}
 
\newcommand{\C}{\mathbb C} 

\newcommand{\R}{\mathbb R}

\newcommand{\bz}{\overline{z}}

\numberwithin{equation}{section}
\makeatletter\@addtoreset{equation}{section}
\usepackage[top=2.5cm, bottom=2.5cm, left=2.5cm, right=2.5cm]{geometry}

\usepackage{color}
\usepackage{colortbl}
\usepackage{tikz-cd}
\usepackage{tikz}
\begin{document}
	

	\title[Poly-meromorphic It\^o--Hermite functions]{Poly-meromorphic It\^o--Hermite functions associated with a singular potential vector on the punctured complex plane}
	
	\author{Hajar  Dkhissi}

				\author{Allal Ghanmi}

			\address{Analysis, P.D.E $\&$ Spectral Geometry, Lab M.I.A.-S.I., CeReMAR, \newline
			Department of Mathematics,
			P.O. Box 1014,  Faculty of Sciences, \newline
			Mohammed V University in Rabat, Morocco}
		\email{hajar${\_}$dkhissi@um5.ac.ma / dkhissihajar@gmail.com} 
	\email{allal.ghanmi@fsr.um5.ac.ma / ag@fsr.ac.ma}


\begin{abstract} 
		We provide a theoretical study of a new family of orthogonal functions on the punctured complex plane solving the eigenvalue problems for some magnetic Laplacian perturbed by a singular vector potential with zero magnetic field modeling the Aharonov--Bohm effect.
		The functions are defined by their $\beta$-modified Rodrigues type formula and extend  the poly-analytic It\^o--Hermite polynomials to the poly-meromorphic setting. Mainly, we derive their different operational representations and give their explicit expressions in terms of special functions. Different generating functions and integral representations are obtained. 

\end{abstract}
	
	
	
%
	\maketitle


\section{Introduction}
The real Hermite polynomials 
\begin{align} \label{Herreal}
H_n(x) =  (-1)^{n}e^{x^{2}}{\frac {d^{n}}{dx^{n}}}(e^{-x^{2}}).
\end{align} 
 are first considered by Laplace 
 \cite{Laplace1810} and next studied in details by Tchebeychev 
   \cite{Chebyshev1859}. They play a crucial role in the one dimensional quantum harmonic oscillator  \cite{FaddeevYakubovskii2009,Lebedev1972}
 and are used to provide analytic proofs for some  combinatorial results  \cite{Dan2009}. Moreover, they have been shown to be useful in expanding the probability generating function of a generalized Poisson distribution \cite{KempKemp1965}. 
A non-trivial two dimensional analog that is not the tensor product of the real Hermite polynomials is giving by the It\^o--Hermite polynomials defined by the Rodrigues formula 
\begin{align} \label{IHP}
	H_{m,n}^\alpha(z, \bz) & =(-1)^{m+n}e^{ \alpha |z|^2 }\dfrac{\partial ^{m+n}} {\partial \bz^{m} \partial z^{n}} \left(e^{- \alpha |z|^2 }\right)
	,  \, \, \alpha >0.
\end{align}
  Here ${\partial}/{\partial z}$ and ${\partial}/{\partial \bz}$ denote the Wirtinger derivatives  with respect to the variables $z$ and $\bz$, respectively. 
%
The polynomials in \eqref{IHP} have been introduced by It\^o within the framework of multiple Wiener integrals  \cite{Ito52}. Since then, they have been extensively studied, namely in connection with various branches of engineering sciences, mathematics and physics \cite{GhJmaa08,Gh13ITSF,Gh2017Mehler,IntInt06,Ismail15a,IsmailTrans2016,Ito52,Matsumoto1996,Shigekawa1987}.
 They
 form an orthogonal complete system in the Hilbert space $L^{2,\alpha}(\C) = L^{2}(\C, e^{-\alpha|z|^2}dxdy)$  (see \cite{IntInt06,GhJmaa08}). Moreover, they provide a concrete  
description of the spectral analysis of the Landau Hamiltonian  \cite{LandauLifshitz1958} 
 \begin{align} 	\label{Hinfty}
		\Delta_{\alpha}  =  \left( i\frac{\partial}{\partial x} - 2\alpha
		y\right) ^2+ \left( i\frac{\partial}{\partial y} +2\alpha x\right) ^2 ; \, \, z=x+iy .
\end{align}
Different generalized class of these polynomials are suggested by special magnetic Schr\"odinger operators. The one associated with constant magnetic field acting on mixed planar automorphic functions attached to a given equivariant pair \cite{ElgourariGh2011} is given by  
      \begin{equation}\label{RodriguezTypeFormula}
 	\mathfrak{G}^{\nu}_{m,n}(z,\bz|\xi) = \frac{(-1)^{m+n}}{\nu^m} e^{\nu|z|^2+\frac \xi 2 z}
 \dfrac{\partial ^{m+n}} {\partial \bz^{m}\partial z^{n}}\left( e^{-\nu|z|^2 -\frac \xi 2z}\right) 
 \end{equation} 
and has been studied in sufficient detail in \cite{GhJmaa08}.
 
In the present paper, we identify a special class of orthogonal and poly-meromorphic functions generalized to previous one and associated with the second order differential operator $$\mathcal{D}_{\alpha,\beta} = \Delta_{\alpha} + S_{\alpha,\beta},$$ which is essentially the Laplacian $\Delta_{\alpha}$ in \eqref{Hinfty}  perturbed by the first order differential operator 
$S_{\alpha,\beta}$ explicitly given by 
\begin{align} 
	S_{\alpha,\beta} 
	&= \frac{\beta}{|z|^2}   \left(  z   \frac{\partial}{\partial z}  -\bz   \frac{\partial}{
		\partial {\bz}}  \right)    
	- \beta   \left( 2\alpha   -\frac{\beta}{|z|^2} \right) 
\end{align}	
and is associated with the closed and singular potential vector 
\begin{align}\label{thetbeta} \widetilde{\theta}_{\beta}(z,\bz)
&=   -\frac{i\beta}{|z|^2} \left( \bz dz - z d\bz \right) , 
\end{align} modeling an Aharonov--Bohm effect. 
Mainly, we aim to explore the role
played by the injection of this singular potential in generating non-trivial orthogonal eigenstates 
within the factorization formalism. 
In fact,  we provide an accurate study of the special functions  \begin{align}\label{Rodbea}
	\psi_{n,m}^{\alpha,\beta}(z,\bz) :  = (-1)^{n}z^{-\beta}e^{\alpha|z|^{2}}\frac{\partial^{n}}{\partial{z^{n}}}(z^{\beta+ m}e^{-\alpha|z|^{2}} ), 
\end{align}  
refereed to as poly-meromorphic (or  $\beta$-modified complex)  It\^o-Hermite  functions.
Here, $\alpha$ and $\beta$ are given fixed reals with $\alpha >0$, and $n$ and $m$ are varying integers such that $n=0,1,2, \cdots$ and $ m >-\beta -1$. This leads to a special generalization of  the complex It\^o-Hermite polynomials in \eqref{IHP}. The polynomial case,  shown in Section \ref{SubsPolymor} to correspond to $m\geq n$ and $\beta \geq 0$, coincides with the polynomials $Z^\beta_{m,n}$ considered by Ismail and Zeng in \cite[Section 3]{IsmailZeng2015}.
More precisely, we are concerned with certain basic algebraic, analytic and  spectral properties of $\psi_{n,m}^{\alpha,\beta}(z,\bz)$ defined by their Rodrigues formula \eqref{Rodbea}. 
In particular, the interrelation  of $	\psi_{n,m}^{\alpha,\beta}(z,\bz) $  with to spacial functions such as the It\^o--Hermite polynomials and the confluent hypergeometric functions are considered in Section \ref{SecPrel}. 
Their spectral realization as eigenfunctions of the magnetic Laplacian $\mathcal{D}_{\alpha,\beta}$ as well as their regularity and their exact bi-order as poly-meromorphic functions on the complex plane   are also studied in Sections  \ref{SubsSpec} and \ref{SubsPolymor}, respectively. 
Associated generating functions are obtained in Section \ref{SecGen}, and next employed to discuss some of their applications such as their integral representations (Section \ref{SecInRep}) and spacial attached integral transforms (Section \ref{SubsIntTrans}).
The fractional side as well as the associated functional spaces of Segal--Bargmann type will be introduced and studied in details in forthcoming papers.


%

 \section{Preliminary results } \label{SecPrel}

It is worth noticing that to avoid multiple-valuedness of the argument $\arg(z)$ in $z^\beta = e^{\beta \log(z)}$ involved in \eqref{Rodbea}, we choose the principal branch for the logarithm.
Notice also that the functions $\psi_{n,m}^{\alpha,\beta}$  satisfy $ \psi_{m,n}^{\alpha,\beta}(z,\bz)=\overline{\psi_{m,n}^{\alpha,\beta}(\bz,z)}$ as well as the symmetry relation   
\begin{eqnarray}\label{R}
	\alpha^m z^{\beta}\psi_{n+\beta,m-\beta}^{\alpha,\beta}(z,\bz)=\alpha^{n+\beta} \bz^{\beta}\overline{\psi_{m,n}^{\alpha,\beta}(z,\bz)} 
\end{eqnarray}
valid for $\beta$ being integer  and for given non-negative integers $m$ and $n$ such that $n \geq max(0,-\beta)$.
%
%
 We have in addition 
		$$\overline{z}^{\beta}\psi_{m,m}^{1,\beta}(z,\bz)=z^{\beta}\psi_{m+\beta,m-\beta}^{1,\beta}(z,\bz) $$
	and 	$$ 
	 \psi_{n,m}^{\alpha,\beta}(z,\bz)=z^{-[\beta]}\psi_{n,m+[\beta]}^{\alpha,\tilde{\beta}}(z,\bz) ,$$
	                           valid for any real $\beta$ such that $\beta+m > -1$, where $[\beta]$ denotes the integer part of $\beta$.
		For the particular case of $\beta =1/2$ we have 
		$$ \sqrt{\alpha |z|}2^{2m+1}\psi_{m,m}^{\alpha,\frac{1}{2}}(z,\bz)=H_{2m+1}(\sqrt{\alpha |z|})=  2(-4)^{m}m! (\alpha |z|)^{1/2} L_{m}^{\left({ {1}/{2}}\right)}(\alpha |z|),$$ 
		where $H_k(x)$ are the Hermite polynomials in \eqref{Herreal} and $L^\alpha_k$ denotes the generalized Lageurre polynomials.
For arbitrary non-negative integer $\beta$, the functions $ \psi_{m,n}^{\alpha,\beta}(z,\bz)$ are closely connected to the complex It\^o-Hermite polynomials in \eqref{IHP} on the punctured complex plane $\C^* =\C \setminus \{0\}$,
\begin{align}\label{RelHermite} z^{\beta}\psi_{n,m}^{\alpha,\beta}(z,\bz)&= \psi_{n,m+\beta}^{\alpha,0}(z,\bz) = \alpha^{\frac{n-m-\beta}{2}}H_{m+\beta,n}(\sqrt{\alpha}z,\sqrt{\alpha}\bz). 
\end{align}
The interrelation with these polynomials for arbitrary $\beta$ can be obtained by specifying $f$ in the Burchnall's operational formula \cite[Proposition 2.3]{Gh13ITSF}
\begin{align*}   (-1)^{n} e^{\alpha|z|^{2}}\frac{\partial^{n}}{\partial{z^{n}}}\left( z^m e^{-\alpha|z|^{2}} f\right)
	& =  \frac{ n! }{ \alpha^m } \sum_{k=0}^{n} \frac{(-1)^{k}  }{k!(n-k)!}
	H_{m,n-k}^\alpha (z,\bz )  \frac{\partial^{k} f}{\partial{z^{k}}}  .
\end{align*}
Thus, for $f(z)= z^{\beta}$ we obtain  
\begin{align*}  
	\psi_{n,m}^{\alpha,\beta} (z,\bz )
	& =  \frac{ n! }{ \alpha^m }  \sum_{k=0}^{n} \frac{(-1)^{k} \Gamma(\beta+1) }{k!(n-k)! \Gamma(\beta-k+1)}
	z^{ -k}  H_{m,n-k}^\alpha (z,\bz )     .
\end{align*}

For the explicit expression of $\psi_{n,m}^{\alpha,\beta}(z,\bz)$ (with $m>-\beta-1$), we claim that 
\begin{align}\label{explicit}
	\psi_{n,m}^{\alpha,\beta}(z,\bz) &  = \sum_{k=0}^{n\wedge ^* (m+\beta)} 
	c_{m,n,k}^{\alpha,\beta}   z^{m-k}\bz^{n-k} ,
\end{align}
where the starred minimum $m\wedge^* b$ is defined by the classical  minimum $m\wedge ^* b = \min(m, b)$ when $b$ is an integer 
and by  $m\wedge ^* b = m$   otherwise. The involved constants $c_{m,n,k}^{\alpha,\beta}$ stand for 
$$c_{m,n,k}^{\alpha,\beta} := \frac{ (-1)^{k} n! \Gamma(\beta+m+1)}{k! (n-k)! \Gamma(\beta + m -k+1)} \alpha^{n-k}.$$ 
The expression in \eqref{explicit} can be handled by applying the Leibniz formula to \eqref{Rodbea} keeping in mind the fact that  
\begin{align*} 
z^{-(\beta+m-k)} \frac{\partial^{k}}{\partial{z^{k}}}(z^{\beta+m})
= \left\{\begin{array}{ll} 0 \, , & \quad \beta =0,1,2, \cdots , k > \beta +m, \\
\displaystyle \frac{\Gamma(\beta+m+1)}{\Gamma(\beta + m -k+1)}  & \quad \mbox{otherwise} . \end{array}\right. 
\end{align*}
Notice also that the monomials $z^q \overline{z}^p  $ can be expressed in terms of the considered functions. In fact, by rewriting them as a derivation of the Gaussian function $e^{-\alpha|z|^{2}}$, we get
\begin{align*}  
	\alpha^p z^q \overline{z}^p  &=  (-1)^p z^q e^{\alpha|z|^{2}}\frac{\partial^{p}}{\partial{z^{p}}}(e^{-\alpha|z|^{2}})\\&=(-1)^p z^q e^{\alpha|z|^{2}}\frac{\partial^{p}}{\partial{z^{p}}}( z^{-\beta-n+p-q} z^{\beta+n+q-p}e^{-\alpha|z|^{2}} )\\	& =\sum_{n=0}^{p} \binom{p}{n}(-1)^{p}\frac{\partial^{p-n}}{\partial{z^{p-n}}}\left(z^{-\beta-n+p-q} \right)z^q e^{\alpha|z|^{2}}\frac{\partial^{n}}{\partial{z^{n}}}(z^{\beta+n+q-p}e^{-\alpha|z|^{2}}).
\end{align*}
The last equality holds making use of the Leibniz formula. Now, by means of the Rodrigues formulas \eqref{Rodbea} with $q-p \geq 0$,  
it follows 
\begin{align} \label{COMP}
	\alpha^p z^q \overline{z}^p 
	&=	\sum_{n=0}^{p} \binom{p}{n}  \frac{\Gamma(\beta +q)}{\Gamma(\beta+q+n-p)} \psi_{n,n+q-p}^{\alpha,\beta}(z,\bz)
\end{align}
for every non-negative integers $p$ and $q$ with $p\leq q$.

The first few terms of $\psi_{n,m}^{\alpha,\beta}$
are given by
 $\psi_{0,m}^{\alpha,\beta}(z,\bz)=z^m  $ and $\psi_{1,m}^{\alpha,\beta}(z,\bz)= z^{m-1} \left( \alpha z \bz-(\beta+m)\right)   $ when $n=0$ and $n=1$ respectively, while for $n=2$ we have 
$$\psi_{2,m}^{\alpha,\beta}(z,\bz)= z^{m-2}\left( \alpha^2 z^2\bz^2 - 2 \alpha (\beta +m) z\bz +  (\beta +m)  (\beta +m-1) \right) .$$
This reveals in particular that the $\psi_{n,m}^{\alpha,\beta}$ are no longer polynomials unless $\beta$ is a non-positive integer. 
This becomes clear from their hypergeometric representation in terms of the hypergeometric functions defined by the series
$${_pF_q}\left( \begin{array}{c} a_1, a_2, \cdots , a_p  \\ c_1, c_2, \cdots , c_q \end{array}\bigg | x \right) = \sum_{n=0}^{\infty }{\frac {(a_1)_{n}(a_2)_{n} \cdots (a_p)_{n} }{(c_1)_{n}(c_2)_{n} \cdots (c_q)_{n}}}{\frac {x^{n}}{n!}} $$
provided that $c_\ell \ne 0, -1, -2, \cdots$ for $\ell =1,2, \cdots,q$.
Indeed, from  \eqref{explicit}  and making use of the classical identities on Gamma function and Pochammer symbol, we get 
  \begin{align} 
  \psi_{n,m}^{\alpha,\beta}(z,\bz) 
   &= \alpha^n z^{m}\bz^{n} \sum_{k=0}^{n\wedge ^* (m+\beta)}    (-n)_k    (-\beta-m)_k    \frac{(-\alpha z\bz)^{-k}}{k!} \nonumber
   \\
   &= \alpha^n z^{m}\bz^{n}    
   {_2F_0}\left( \begin{array}{c} -n , -\beta-m\\ - \end{array}\bigg | -\frac{1}{\alpha |z|^2}  \right) \label{hypergeometric1}\\
   &=  \frac{(-1)^n  \Gamma(\beta+m+1)}{\Gamma(\beta+m-n+1)}    z^{m-n}    
   {_1F_1}\left( \begin{array}{c} -n \\  \beta+m -n +1  \end{array}\bigg |  \alpha|z|^2   \right). \label{hypergeometric2}
  \end{align}
Subsequently, by means of \eqref{hypergeometric2} we  obtain the expression of $\psi_{n,m}^{\alpha,\beta}$ in terms of the generalized Laguerre polynomials   \cite[p. 240]{Magnus1966},
to wit 
 \begin{align} \label{ExpLag}
\psi_{n,m}^{\alpha,\beta}(z,\bz) 
&= 
(-1)^n n!   z^{m-n}  L_{n}^{(\beta+m -n )}( \alpha|z|^2  ). 
\end{align}
This can also be recovered starting from \eqref{Rodbea} using the derivation formula 
 \cite[p.241]{Magnus1966}.
Accordingly, with the use of \eqref{ExpLag}, it is straightforward to prove 
the orthogonality of $\psi_{n,m}^{\alpha,\beta}$ in the Hilbert space   
$L^{2,\alpha}_{\beta}(\mathbb{C}):
=L^{2}(\mathbb{C}, 
d\mu_{\alpha,\beta})$  of square integrable functions with  respect to the measure $d\mu_{\alpha,\beta}(z) = |z|^{2\beta}e^{-\alpha|z|^{2}}  d\lambda(z) $, 
where 
$d\lambda(z)= dxdy$ denotes the Lebesgue measure on the complex plane with $z=x+iy$, $x,y\in \R$. More explicitly, we have
$$\int_{\C} \psi_{n,m}^{\alpha,\beta}(z,\bz) \overline{\psi_{k,j}^{\alpha,\beta}(z,\bz)} 
|z|^{2\beta}e^{-\alpha|z|^{2}}  d\lambda(z) 
= \frac{\pi \alpha^n  n!}{\alpha^{m+\beta+1}}\Gamma(\beta+m+1) \delta_{m,j} \delta_{n,k}.$$ 
We conclude these preliminaries by noticing that starting from \eqref{Rodbea}, it is clear that 
$$ z \psi_{n,m}^{\alpha,\beta}(z,\bz) = \psi_{n,m+1}^{\alpha,\beta-1}(z,\bz).$$
Also, by rewriting $\partial^{n+1}$ as $\partial^{n}\partial$, 
one obtains the
recurrence formula 
\begin{eqnarray}\label{RecForm} \alpha \bz \psi_{n,m}^{\alpha,\beta}(z,\bz) = \psi_{n+1,m}^{\alpha,\beta}(z,\bz) + (\beta+m) \psi_{n,m-1}^{\alpha,\beta}(z,\bz).
\end{eqnarray}
Additional recurrence formulas can be derived from the different known ones for the generalized Laguerre polynomials. For example from those in \cite[p.  241]{Magnus1966} one obtains
\begin{align*}
(1) \quad	&\psi_{n,m+1}^{\alpha,\beta-1}(z,\overline{z})
	=\psi_{n,m+1}^{\alpha,\beta}(z,\overline{z})+n\psi_{n-1,m}^{\alpha,\beta}(z,\overline{z}) 
	\\
	(2)  \quad &\psi_{n+1,m+1}^{\alpha,\beta}(z,\overline{z})=\left(\alpha|z|^2 -[n+\beta+m+1] \right)\psi_{n,m}^{\alpha,\beta}(z,\overline{z})- n\left( \beta+m \right)\psi_{n-1,m-1}^{\alpha,\beta}(z,\overline{z})
	\\
	(3) \quad & 
	\alpha \overline{z} \psi_{n,m+2}^{\alpha,\beta}(z,\overline{z}) = \left(\alpha|z|^2 -n \right)\psi_{n,m+1}^{\alpha,\beta}(z,\overline{z}) - n\left( \beta+m+1 \right)\psi_{n-1,m}^{\alpha,\beta}(z,\overline{z})
	\\
	(4) \quad & \psi_{n+1,m+1}^{\alpha,\beta}(z,\overline{z}) = \left(\alpha|z|^2 -n-1 \right)\psi_{n,m}^{\alpha,\beta}(z,\overline{z}) - z\left( \beta+m\right)\psi_{n,m-1}^{\alpha,\beta}(z,\overline{z})  \\ 
(5) \quad	& 	\left(\beta+m-n+\alpha|z|^2 \right)\psi_{n,m}^{\alpha,\beta}(z,\overline{z}) = z\left( \beta+m\right)\psi_{n,m-1}^{\alpha,\beta}(z,\overline{z})	+ \alpha\overline{z}	\psi_{n,m+1}^{\alpha,\beta}(z,\overline{z}).
\end{align*} 

\section{Spectral realization}  \label{SubsSpec}

The result below shows that the functions $\psi_{m,n}^{\alpha,\beta}$ are $L^{2}$-eigenfunctions of the perturbed magnetic Lapalcian defined by 
\begin{align} \label{eigVP1}
	\Delta_{\alpha,\beta} :=	- \frac{\partial ^2}{\partial z
		\partial {\bar z}} +  \left(  \alpha  -\frac{\beta}{|z|^2} \right) \bz   \frac{\partial}{
		\partial {\bz}}  
\end{align} 
and the second order differential operator 
\begin{eqnarray} \label{Lap2} 
	\widetilde{\Delta}_{\alpha,\beta} :=  -  \frac{\partial^2}{\partial z \partial \bz } + \alpha z \frac{\partial}{\partial z} - \frac{\beta}{z} \frac{\partial}{\partial \bz} . 
\end{eqnarray}

\begin{theorem}  \label{thm1}
	The functions $	\psi_{n,m}^{\alpha,\beta}$ 
	form an orthogonal system in  $L^{2,\alpha}_{\beta}(\mathbb{C})$  that solve the eigenvalue problems $\Delta_{\alpha,\beta} =\alpha n$ and $\widetilde{\Delta}_{\alpha,\beta} =    \alpha m $.
\end{theorem}

\begin{proof}
	Notice first that the second order differential operator in \eqref{eigVP1} can be rewritten as
	$$	\Delta_{\alpha,\beta} 
	=	A^{*_{\alpha,\beta}} A    =  A^{} A^{*_{\alpha,\beta}} - \alpha,$$
	where $A^{} = {\partial}/{\partial{\bz}}$ and $A^{*_{\alpha,\beta}}$ are the first order differential operators given by
	\begin{align} \label{Aadjoint}
		 A^{*_{\alpha,\beta}} = -[\rho_{\alpha,\beta} (z)]^{-1} \frac{\partial}{\partial{z}}\left( \rho_{\alpha,\beta} (z) f(z)\right).
	\end{align}
	Here $\rho_{\alpha,\beta} (z):= |z|^{2\beta} e^{-\alpha |z|^2}$.
	Thus, using the commutation rule
	$
	A^{}A^{*_{\alpha,\beta}} - A^{*_{\alpha,\beta}} A^{} =\alpha Id$,  one proceeds by induction to get the identity
	$
	A^{}(A^{*_{\alpha,\beta}})^{n+1}=
	(A^{*_{\alpha,\beta}})^{n+1}A^{}+\alpha(n+1) (A^{*_{\alpha,\beta}})^{n}$. Subsequently, we have
	$$ 
	\Delta_{\alpha,\beta}  ((A^{*_{\alpha,\beta}})^{n}(g))   =
	\left( A^{} A^{*_{\alpha,\beta}} - \alpha\right)  ((A^{*_{\alpha,\beta}})^{n}(g)) = \alpha n ((A^{*_{\alpha,\beta}})^{n}(g)) ,$$
	for any $g\in \ker (A^{})$. Thus,  by considering the case of the generic elements $g_m(z)= z^m$ with $m\in \Z$ for $z$ in the punctured complex plane,  one deduces that the function 
	\begin{align}\label{Rodbea0}
		\psi_{n,m}^{\alpha,\beta}(z,\bz) :=(A^{*_{\alpha,\beta}})^{n}(g_m) =
		(-1)^{n} [\rho_{\alpha,\beta} (z)]^{-1} \frac{\partial^{n}}{\partial{z^{n}}} \left( z^m  \rho_{\alpha,\beta} (z)  \right) 
	\end{align}   
	is a eigenfunction of $ \Delta_{\alpha,\beta} $ with $n\alpha$ as a corresponding eigenvalue. 
	  
 Now, to prove that $\widetilde{\Delta}_{\alpha,\beta} \psi_{n,m}^{\alpha,\beta}  =    \alpha m \psi_{n,m}^{\alpha,\beta} $, we make use of 
 the partial raising operations 
 \begin{eqnarray}\label{Raisingn}  - \left( \frac{\partial}{\partial z}  - \alpha \bz  + \frac{\beta}{z} \right) \psi_{n,m}^{\alpha,\beta}(z,\bz)= \psi_{n+1,m}^{\alpha,\beta}(z,\bz) 
 \end{eqnarray}
 and
 \begin{eqnarray}\label{Raisingm} -\frac{1}{\alpha} \left( \frac{\partial}{\partial \bz} - \alpha z \right) \psi_{n,m}^{\alpha,\beta}(z,\bz) = \psi_{n,m+1}^{\alpha,\beta}(z,\bz)  ,
 \end{eqnarray}
 which are immediate by straightforward computation. 
On the other hand, from \eqref{RecForm} and \eqref{Raisingn} one has
		\begin{eqnarray}\label{Ann2}  \left( \frac{\partial}{\partial z} + \frac{\beta}{z} \right) \psi_{n,m}^{\alpha,\beta}(z,\bz) =  (\beta+m) \psi_{n,m-1}^{\alpha,\beta}(z,\bz).
		\end{eqnarray}
		Accordingly, by combining \eqref{Raisingm} and \eqref{Ann2}, it follows 
		\begin{eqnarray}  \left( \frac{\partial}{\partial z} + \frac{\beta}{z} \right)\left( \frac{\partial}{\partial \bz} - \alpha z \right) \psi_{n,m}^{\alpha,\beta}(z,\bz)
			=  - \alpha (\beta+m+1) \psi_{n,m}^{\alpha,\beta}(z,\bz).
		\end{eqnarray}
		This completes the proof by observing that the operator $\widetilde{\Delta}_{\alpha,\beta}$ in \eqref{Lap2} can be factorized, up to an additive constant, as 
		$$ \widetilde{\Delta}_{\alpha,\beta} = - \left( \frac{\partial}{\partial z} + \frac{\beta}{z} \right)\left( \frac{\partial}{\partial \bz} - \alpha z \right) - \alpha(\beta+1) .$$
\end{proof}

\begin{remark}   
Let $ E =  z \partial/{\partial z}$ be the Euler derivative operator and  $\overline{E} = \bz \partial/{\partial \bz}$ its complex conjugate.  Then, the functions $\psi_{n,m}^{\alpha,\beta}$ satisfy 
	$$ (E - \overline{E} ) \psi_{n,m}^{\alpha,\beta} = (m-n) \psi_{n,m}^{\alpha,\beta},$$ which readily follows since $ \widetilde{\Delta}_{\alpha,\beta}  - \Delta_{\alpha,\beta} =  \alpha( E - \overline{E})$. This 
	is the analog for $\psi_{n,m}^{\alpha,\beta}$  at
	arbitrary integer $m$ such that $m > -\beta-1$ of the one obtained in  \cite{IsmailZeng2015}.
\end{remark}

\begin{remark}
	The orthogonality of $\psi_{n,m}^{\alpha,\beta}$ in $L^{2,\alpha}_{\beta}(\mathbb{C})$ can be reproved by observing that the operator  $A^{*_{\alpha,\beta}}$ is in fact
	the adjoint of $A^{}$ when acting on a densely domain in $L^{2,\alpha}_{\beta}(\mathbb{C})$.
	In fact, we have
	\begin{align*}
		\scal{ \psi_{n,m}^{\alpha,\beta},\psi_{n+p,q}^{\alpha,\beta} }_{\alpha,\beta} 
		&= \scal{A^{}(A^{*_{\alpha,\beta}})^{n}(g_m), (A^{*_{\alpha,\beta}})^{n+p-1}(g_q)}_{\alpha,\beta} \\
		&= \scal{(A^{*_{\alpha,\beta}})^{n}A^{}(g_m)+n\alpha (A^{*_{\alpha,\beta}})^{n-1}(g_m),(A^{*_{\alpha,\beta}})^{n+p-1}(g_q)}_{\alpha,\beta}  	\\
		&
		= \alpha n \scal{\psi_{n-1,m}^{\alpha,\beta},\psi_{n+p-1,q}^{\alpha,\beta} }_{\alpha,\beta}\\&
		=  \alpha^{n}   n! \scal{\psi_{0,m}^{\alpha,\beta},\psi_{p,q}^{\alpha,\beta} }_{\alpha,\beta}
		\\&	= \frac{\pi  \alpha^n n! }{\alpha^{m+\beta+1}}\Gamma(\beta+m+1)
		\delta_{0,p}\delta_{m,q}.
	\end{align*}
	
	
	%
	\end{remark}

Below, we prove that the considered functions are closely connected to the spectral analysis of a specific Schr\"odinger operator  $ L  = \nabla_{\theta} ^* \nabla_{\theta}$ associated with a specific singular vector potential $\theta$, where $\nabla_{\theta}  = d   + i \bold{ext}_\theta  $ is the co-derivation operator  
acting on $\Omega^\infty_{p,c}(\C)$,  the space of smooth  differential $p$-forms with compact support, and $\bold{ext}_\theta (\omega:)= \theta \wedge \omega $ denotes the exterior multiplication by $ \theta$.
The operator $\nabla^{*}_\theta$ denotes its formal adjoint with respect to the Hermitian scalar product
\begin{eqnarray}
	\label{ps}\Big(\omega_1 , \omega_2\Big)_p&=&\int_{\C}\omega_1 \wedge
	\star \omega_2 ,
\end{eqnarray}
for $\omega_1, \omega_2 \in \Omega^\infty_{p,c}(\C)$.
Here $\star$ is the Hodge star operator on differential forms defined  to satisfy $\star (f \omega)= \overline{f}(\star \omega)$ for scalar functions $f$, and  
$\star (dz\wedge d\bar z) =2i$.
This readily follows for the metric $ds^2$ being conformal to the Euclidean metric $ds^2(z) = dz \otimes d\bz$.
%
Now, by considering the strong extensions of the differential operators   $d$, $\nabla$ and $ L$  initially defined on  $\mathcal{C}^\infty_0(\mathbb{C}) = \Omega^\infty_{0,c}(\C)$, 
we can extend them to the whole $L^2$-Hilbert space as the closure of the $L^2$-norm with respect to $\scal{\omega_1 , \omega_2}_0$. 

\begin{lemma}
	\label{lemLLExplicit} 
	For given reals $\alpha$ and $\beta$ such that $\alpha>0$, we set  $\nabla_{\alpha,\beta} =d + i \bold{ext}_{\theta_{\alpha,\beta}}$ and we let $ \theta_{\alpha,\beta} $ be the real-valued differential $1$-form given by
	$ \theta_{\alpha,\beta}   := -i(\partial - \overline{\partial}) Log (\rho_{\alpha,\beta}) $. Then, $
	\nabla_{\alpha,\beta}^* \nabla_{\alpha,\beta}
	$
	coincides with the second order differential operator given by \begin{align}\label{explicitLG}
		\mathcal{D}_{\alpha,\beta}
		&=-  
		\left\{ \frac{\partial ^2}{\partial z
			\partial {\bar z}} +  \left(  \alpha  -\frac{\beta}{|z|^2} \right) \left(  z   \frac{\partial}{\partial z}  -\bz   \frac{\partial}{
			\partial {\bz}}  \right)    \right\}
		+  \left( \alpha   -\frac{\beta}{|z|^2} \right) ^2 |z|^2.
	\end{align}	
\end{lemma}

\begin{proof}
	Notice first that the differential 1-form $ \theta_{\alpha,\beta} $ is explicitly given by 
	$ \theta_{\alpha,\beta}  = 
	i	k^{\alpha}_{\beta}(z) \left( {\bar	{z}}dz - zd\bar z \right) 
	$ with  
		$	k^{\alpha}_{\beta}(z): =  \alpha - {\beta}/{|z|^2}.
		$	 	 
	Straightforward computation
	making use of the well-known facts $d^*=-\star d \star$ and  $(\bold{ext}_{\theta_{\alpha,\beta}})^*= \star (\bold{ext}_{\theta_{\alpha,\beta}}) \star$  
	shows that for every smooth differential $1$-form $\omega = A dz + B d\bz$ we have
	\begin{align}\label{actiond}
		d^*\left( A dz + B d\bz\right) 
		= -2   \left( \frac{\partial A}{\partial \bz} +\frac{\partial B}{\partial z} \right)   
	\end{align}
	and 
	\begin{align}\label{actiontheta}
		(\bold{ext}_{\theta_{\alpha,\beta}})^*\left( A dz + B d\bz\right) 
		= -2i   k^{\alpha}_{\beta}(z)
		\left( z A - \bz B\right) .
	\end{align}
	Therefore, by taking $\omega = df = \partial f dz + \overline{\partial} f d\bz$ in \eqref{actiond} one recovers   the explicit expression of the Hodge--de Rham operator  	  \begin{align}\label{Act1}
	\nabla_{0,0}^* \nabla_{0,0}=	 \frac 1 4  d^*d =
		-  \frac{\partial ^2}{\partial z\partial {\bar z}}.
	\end{align}
	Moreover, the explicit differential expression of the operators 
	$d^*(\bold{ext}_{\theta_{\alpha,\beta}})$, $ (\bold{ext}_{\theta_{\alpha,\beta}})^* d $
	and $(\bold{ext}_{\theta_{\alpha,\beta}})^*(\bold{ext}_{\theta_{\alpha,\beta}}) $ are given respectively by
	\begin{align}
		& \qquad  d^*(\bold{ext}_{\theta_{\alpha,\beta}}) f= 2i   k^{\alpha}_{\beta}(z)
		(E-\bar	E) f
		, \label{Act2}\\ 
		& \qquad(\bold{ext}_{\theta_{\alpha,\beta}})^*d	f   = -2i   k^{\alpha}_{\beta}(z)
		(E-\bar E)f , \label{Act3}
		\\ 	
		&\qquad
		(\bold{ext}_{\theta_{\alpha,\beta}})^*(\bold{ext}_{\theta_{\alpha,\beta}}) f = 4 
		k^{\alpha}_{\beta}(z)^2 |z|^2f. \label{Act4}	
	\end{align}
	Subsequently, by expanding $\nabla_{\alpha,\beta}^* \nabla_{\alpha,\beta}$ as
	\begin{eqnarray} 
		\nabla_{\alpha,\beta}^* \nabla_{\alpha,\beta}
		= \frac 14 \left\{d^*d +i (d^*
		\bold{ext}_{\theta_{\alpha,\beta}}-(\bold{ext}_{\theta_{\alpha,\beta}})^*d)
		+ (\bold{ext}_{\theta_{\alpha,\beta}})^*(\bold{ext}_{\theta_{\alpha,\beta}})\right\} ,
	\end{eqnarray} 
	and next making use of \eqref{Act1}--\eqref{Act4}, we get its  explicit expression   given through $
	 \mathcal{D}_{\alpha,\beta}$ in \eqref{explicitLG}.
\end{proof}

\begin{remark} 
The  second order differential operator $	\mathcal{D}_{\alpha,\beta}$ in 
\eqref{explicitLG}  is a magnetic Laplacian with a constant 
  homogeneous magnetic field of magnitude $\alpha$ applied perpendicularly on the complex plane. Indeed,  we have 
	$$ d \theta_{\alpha,\beta} =d \theta_{\alpha}  = 2i\partial\overline{\partial}( Log (\rho_{\alpha,\beta})) = 2i\alpha dz\wedge d\bz,$$
	where $  \theta_{\alpha,\beta}   =  \theta_{\alpha} + \widetilde{\theta}_{\beta}
	$ with  
	$ \widetilde{\theta}_{\alpha} =i	  \alpha   \left( {\bar	{z}}dz - zd\bar z \right)$ and  	 
	$ \widetilde{\theta}_{\beta}= 
	-i   \beta \left( zdz - zd\bar z \right)/|z|^2$.
		Moreover, the operator $	\mathcal{D}_{\alpha,\beta}$  is essentially the classical Landau Hamiltonian in \eqref{Hinfty} 
	perturbed by a first order differential operator associated with the potential $1$-form $\widetilde{\theta}_{\beta}$ 
	closed (with zero magnetic field), singular (at the origin) and modeling the Aharonov--Bohm effect. 
\end{remark}

\begin{theorem}  \label{corcor}
	The functions $ |z|^{2\beta} e^{-\alpha|z|^2}	\psi_{n,m}^{2\alpha,2\beta}$  are  $L^{2}$-eigenfunctions of the magnetic Laplacian $	\nabla_{\alpha,\beta}^* \nabla_{\alpha,\beta}$
	with	  $ \alpha  (2 n + 1) $ as associated eigenvalue.
\end{theorem}

\begin{proof}
	The proof is immediate using Theorem \ref{thm1}, Lemma \ref{lemLLExplicit} and observing that the operators   $\Delta_{2\alpha}^{2\beta} $ in \eqref{eigVP1} and the magnetic Laplacian $\mathcal{D}_{\alpha,\beta}$ in \eqref{explicitLG}
	are unitary equivalent. More precisely, we have 
	$$ \rho_{\alpha,\beta} \left(\Delta_{2\alpha}^{2\beta}  +  \alpha   \right) \left(   (\rho_{\alpha,\beta})^{-1} f \right)  =       \mathcal{D}_{\alpha,\beta}  ,$$
	which readily follows since 
	 $$ \mathcal{D}_{\alpha,\beta}  = 	B^{*_{-\alpha,-\beta}} \circ 	A^{*_{\alpha,\beta}} + \alpha 
	=   	A^{*_{\alpha,\beta}} \circ   B^{*_{-\alpha,-\beta}} - \alpha.$$
	Here $ 	A^{*_{\alpha,\beta}}$  is as in \eqref{Aadjoint} and $B^{*_{\alpha,\beta}} $ is the differential operator given by  
	$$
	B^{*_{\alpha,\beta}} f= [\rho_{\alpha,\beta} (z)]^{-1} \frac{\partial}{\partial{\bz}}\left( \rho_{\alpha,\beta} (z) f \right) 
	.$$
\end{proof}


\section{Analytical side (poly-meromorphy)} \label{SubsPolymor} 
In this section, we discuss the regularity of $\psi_{n,m}^{\alpha,\beta}(z,\bz)$ as poly-meromorphic functions on the complex plane and we determinate the "bi-order" of its unique  pole. Recall first from  
\cite[p 199]{Balk1997} that a $n$-meromorphic function (or  poly-meromorphic of order $n$) on an open set $U \subset \C$ is a complex-valued function for which there exist some meromorphic functions $\psi_k$; $k=0,1, \cdots,n-1$ on $U$ such that 
\begin{align}\label{npolymer}
	f(z,\bz) &=   \psi_{0}(z)+  \bz  \psi_{1}(z) + \cdots + \bz^{n-1} \psi_{n-1}(z) .
\end{align}
They are called simply $n$-analytic ($n$-poly-holomorphic) when the component functions are holomorphic in $U$, $\psi_k\in Hol(U)$. The latter ones can equivalently be defined as those satisfying the generalized Cauchy--Riemann equation $\partial^n /\partial{\bz}^n = 0$.
%
In order to give the exact statement of the main result of this section, we need first to precise the notion of bi-order of  a zero or a pole of a given poly-meromorphic function on $\C$. 
Thus, for a given non-constant $n$-analytic function $f$ on an open set $U\subseteq \mathbb{C}$, a point $z_0\in U$ is said to be a zero of bi-order $(r,s)$, for given  non-negative integers $r, s$ with $0\leq r \leq n-1$  and $(r,s) \ne (0,0)$, if the following conditions are met
\begin{enumerate}
	\item[(a)] $f$ can be rewritten as $f=\overline{\left( z-z_0\right)}^s g$
	for certain $(n-s)$-analytic function \begin{eqnarray}\label{fctReg}
		g=\sum_{k=0}^{n-s-1}\overline{(z-z_0)}^k \phi_k,
	\end{eqnarray}
	with $\phi_k \in Hol(U)$ and $\phi_0$ is not identically zero on $U$. 
	\item[(b)]  $z_0$ is a zero of order $r$ for the constant component function $\phi_0$ in \eqref{fctReg}.
\end{enumerate}  
The first condition $(a)$ is to say that $z_0$ is a zero of order $s$ for $f$ seen as a polynomial in $\bz$.
Notice also that the suggested definition is equivalent to have 
\begin{eqnarray}\label{fctRegequiv}
	f\left( z,\overline{z}\right)=\overline{\left(z-z_0\right)}^s \left(\left( z-z_0\right)^r \varphi_0+\sum_{k=1}^{n-s-1}\left( \overline{z-z_0}\right)^k \phi_k \right) 
\end{eqnarray} 
for $\varphi_0, \phi_j \in Hol(U)$ with $\varphi(z_0)\neq 0$.
Notice here that $z_0$ does not need to be a zero of the holomorphic components $\phi_k$; $k=1,2, \cdots, n-s-1$. 
However, for the particular case of $z_0$ being a common zero of $\phi_k$ the expression in \eqref{fctRegequiv} reduces to 
$$ f(z,\bz) = (z-z_0)^r\overline{(z-z_0)}^s g(z,\bz),$$ 
for certain non-vanishing poly-analytic function $g$. This makes  $z_0$ a zero of $f$ of  bi-order $(r,s)$.

A point $z_0 \in U$ is said to be a pole of order $r$ ($r < 0$) for given $n$-poly-meromorphic function  $f$  in \eqref{npolymer}
if $\left( z-z_0 \right)^{|r|} f$ is a $n$-analytic function on $U$ and $r$ is the smallest negative integer satisfying this property.  This is equivalent to $z_0 $ being a pole  for certain component meromorphic function $\psi_j$ with
$$r=min\{ Ord_p \left(z_0,\psi_j \right), j=0,1,\cdots, n-1\},$$
where $Ord_p\left(z_0,\psi_j \right)$ is exactly the multiplicity of $z_0$ if it is a pole of $\psi_j$ and $0$ otherwise.
Such pole is said to be of bi-order $(r,s)$, if in addition $(a)$ is satisfied.
Accordingly, we denote by $\mbox{bi-Ord}(z_0;f) $ the bi-order of a point $z_0$ when is a zero or a pole of given $n$-poly-meromorphic function $f$.

\begin{theorem}
	The functions $\psi_{n,m}^{\alpha,\beta}(z,\bz)$ are poly-meromorphic on $\C$. The origin is  either a zero or a pole of bi-order 
	\begin{align}\label{orderpole} \mbox{bi-Ord}(0; \psi_{n,m}^{\alpha,\beta}) = \left( m-[n \wedge^*(\beta+m)], n-[n \wedge^*(\beta+m)]\right) .
	\end{align} 
\end{theorem}

\begin{proof}
	First of all, we point out that in view of \eqref{explicit}, it is clear that the terms $\bz^{n-k}$ are always regulars for every $k \leq n\wedge ^* (m+\beta)\leq n$. The singularity of $\psi_{n,m}^{\alpha,\beta}(z,\bz)$ then lies in $z^{m-k}$ for $k \leq n\wedge ^* (m+\beta)$.
	In particular, the functions $\psi_{n,m}^{\alpha,\beta}(z,\bz) $ are poly-holomrphic  (since they are polynomials in $z$ and $\bz$) if and only if
	$m \geq n\wedge ^* (m+\beta)$. The latter condition is equivalent to 
	$\beta$ being a non-positive integer or $m \geq n$. 
	In this case the expression of $\psi_{n,m}^{\alpha,\beta}(z,\bz) $ reduces to 
	\begin{align*}
		\psi_{n,m}^{\alpha,\beta}(z,\bz) 
		& =  z^{m-[n \wedge^*(\beta+m)]}  \bz^{n-[n \wedge^*(\beta+m)]} R_{n,m;n \wedge^*(\beta+m) }^{\alpha,\beta}(z,\bz) ,
	\end{align*}
	where the involved $R_{n,m; N}^{\alpha,\beta}(z,\bz) $ are the radial polynomials given by
	\begin{align}
		\label{polR} R_{n,m;N }^{\alpha,\beta}(z,\bz) &:=  \sum_{k=0}^{N}  c_{m,n,k}^{\alpha,\beta} |z|^{2([n \wedge^*(\beta+m)]-k)}.
	\end{align}
	Subsequently, since $ \beta+m +1 > n \wedge^*(\beta+m)$ and then  $c_{m,n,n\wedge ^* (m+\beta)}^{\alpha,\beta} \ne 0 $, the origin is a zero of $
	\psi_{n,m}^{\alpha,\beta}(z,\bz)$ whenever $\min(m,n) > n \wedge^*(\beta+m) $. Its bi-order is then 
	\begin{align*} 
		\mbox{bi-Ord}(0; \psi_{n,m}^{\alpha,\beta}) 
		& = 
		\left\{\begin{array}{ll} (m-n , 0) ,     & \mbox{if} \quad \beta \in \R \setminus \Z^{-} \mbox{ or } m\geq n , \\
			(-\beta, n-m-\beta) ,     & \mbox{if} \quad \beta \in \Z^{-}, \beta \ne 0 \mbox{ and } n > \beta +m . \end{array}\right.
	\end{align*}
	To achieve the proof, it remains sufficient to discuss the case of 
		$m < n\wedge ^* (m+\beta)$ (i.e. $n > m > -\beta -1$ and $\beta \notin \Z^-$). In this case we have 
		\begin{align*}
			\psi_{n,m}^{\alpha,\beta}(z,\bz) 
			&  =  \sum_{k=0}^{m}  c_{m,n,k}^{\alpha,\beta} z^{m-k}  \bz^{n-k}  +   \sum_{k=m+1}^{n \wedge^*(\beta+m)}  c_{m,n,k}^{\alpha,\beta} z^{m-k}  \bz^{n-k}  \\
			&  =  \bz^{n-m}  R_{n,m;m }^{\alpha,\beta}(z,\bz)     +   \bz^{n-[n \wedge^*(\beta+m)]} \sum_{j=0}^{[n \wedge^*(\beta+m)]-m-1}  c_{m,n,m+1+j}^{\alpha,\beta} \frac{\bz^{[n \wedge^*(\beta+m)]-m-1-j} }{z^{j+1}}	\\
			&  =    \bz^{n-m}  R_{n,m;m }^{\alpha,\beta}(z,\bz)     +   \frac{\bz^{n-[n \wedge^*(\beta+m)]}}{z^{[n \wedge^*(\beta+m)]-m}} S_{n,m; [n \wedge^*(\beta+m)]-m-1 }^{\alpha,\beta}(z,\bz)      ,
		\end{align*} 
		where $R_{n,m;m }^{\alpha,\beta}(z,\bz)  $  is as in \eqref{polR} with $N=m$, and 
		$$ S_{n,m; N }^{\alpha,\beta}(z,\bz)   =     \sum_{j=0}^{[n \wedge^*(\beta+m)]-m-1}  c_{m,n,m+1+j}^{\alpha,\beta} |\bz|^{2([n \wedge^*(\beta+m)]-m-j-1)} .$$
		It convenient to mention here that both $R_{n,m;N }^{\alpha,\beta}(z,\bz)  $  and 
		$ S_{n,m; N }^{\alpha,\beta}(z,\bz)$ are  poly-analytic radial polynomials on the whole complex plane for which the origin is not a zero (for again $c_{m,n,n \wedge^*(\beta+m)}^{\alpha,\beta} \ne 0$). 
		This proves that the functions $\psi_{n,m}^{\alpha,\beta}$
		are purely poly-meromorphic functions with $0$ as unique pole if and only if $\beta \notin \Z^{-}$  and $m<n$. The multiplicity of their unique pole is given by 
			$	\mbox{Ord} (0; \psi_{n,m}^{\alpha,\beta})
			= m - [n \wedge^*(\beta+m)] < 0$ and then
			\begin{align*} 
				\mbox{bi-Ord}(0; \psi_{n,m}^{\alpha,\beta}) & = 
				\left\{\begin{array}{ll} \left( m-n, 0\right)  & \quad \beta \in \R \setminus \Z   ,  n > m , \\
					\left( m-n,0\right)     & \quad \beta \in \Z^{+*},  \beta+ m  \geq n > m , \\
					\left( -\beta, n-\beta-m\right)     & \quad \beta \in \Z^{+*}, n \geq \beta+ m > m  ,\\
					\left( -\beta, n-\beta-m\right)     & \quad \beta \in \Z^{-*}, n > m . \end{array}\right.
			\end{align*}
		\end{proof}


		%
		
		\begin{remark}
			The polynomial case, i.e., the restriction to the case of $m\geq n$ (with $\beta \geq 0$) leads to the class of polynomials $Z^\beta_{m,n}(z,w)$ introduced and studied by Ismail and Zeng \cite[Section 3]{IsmailZeng2015}. Some of the obtained results in the previous section generalize the one derived in \cite[Section 3]{IsmailZeng2015}.
		\end{remark}

\section{Generating functions } \label{SecGen}

Notice first that using the relation of  $\psi_{n,m}^{\alpha,\beta}$ to the generalized Laguerre polynomials and the generating function for  the latter ones \cite{Brychkov, Magnus1966},  one obtains
$$\frac{(-1)^{n}}{n!j!}z^{j+n}\psi_{n,m}^{\alpha,\beta}(z,\overline{z})=\sum_{k=0}^{n}\frac{(-1)^k z^k}{k!(n-k)!(j-(n-k))!}\psi_{k,j}^{\alpha,\beta}(z,\overline{z})$$
for all $ j>max(-\beta,n)$. 
Moreover, by \cite[p 242]{Magnus1966} with $\beta+k > -1 $, we have 
$$\sum_{n=0}^{+\infty} \frac{ t^n \psi_{n,n+k}^{\alpha,\beta}(z,\overline{z})\psi_{n,n+k}^{\alpha,\beta}(w,\overline{w})}{n!(1+\beta+k)_n} =\frac{(wz)^k}{(1-t)^{\beta+k+1}}e^{\frac{\alpha t(|z|^2+|w|^2)}{(1-t)}}  {_0F_1}\left( \begin{array}{c}-\\  \beta+k+1  \end{array}\bigg |  \frac{|\alpha zw|^2}{t(1-t)^2} \right).
$$
The next one is an analog of the standard one for the It\^o--Hermite polynomials \cite[p. 7]{Gh13ITSF}
  which  appears as the special case when $\beta =0$ and  $\alpha=1$.

\begin{theorem}\label{HajarGen1}
	For any real $\beta >-1$,
	the functions $\psi_{n,m}^{\alpha,\beta} $ satisfy  	
	\begin{align}\label{genfctD}
		\sum_{m,n=0}^{+\infty}\frac{u^m v^n}{m!n!}\psi_{n,m}^{\alpha,\beta}(z,\bz) =\left( 1-\frac{v}{z}\right) ^{\beta}
		e^{ z u  +\alpha v  \bz - u v } .
	\end{align} 
	
\end{theorem}

\begin{proof} Starting from the left hand-side of \eqref{genfctD} and inserting \eqref{Rodbea} and next interchanging the sum in $m$ and the $n$-th derivative, one obtains
	\begin{align*}
		\sum_{m,n=0}^{+\infty}\frac{u^m v^n}{m!n!} \psi_{n,m}^{\alpha,\beta}(z,\bz) &
		= z^{-\beta}e^{\alpha|z|^2} \sum_{n=0}^{+\infty} \frac{(-v)^n}{n!} \frac{d^n \left( \varphi_{\bz}(x)\right) }{d{x^n}}|_{x=z} 
		\\&=  z^{-\beta}e^{\alpha|z|^2} \varphi_{\bz}(z-v) ,
	\end{align*} 
	with $ \varphi_{\bz}(x) := z^{\beta} e^{-\alpha (\bz+u)x}$. The last equality follows using the translation operator of the Taylor series of the involved function and gives rise to the right hand-side of \eqref{genfctD}. 
\end{proof}

The following results are partial generating functions for $\psi_{n,m}^{\alpha,\beta}$ (with fixed $n$ or $m$). 

\begin{proposition}\label{HajarGen3}
For $\beta>-1$, we have
 \begin{align}\label{GA}
	\sum_{n=0}^{+\infty}\psi_{n,k}^{\alpha,\beta}(z,\bz)\frac{ v^n}{n!}=\frac{(z-v)^{k+\beta}}{z^{\beta}}e^{\alpha v \bz} \end{align}
	as well as
	\begin{align}\label{GA5}
		\sum_{m=0}^{+\infty}
		\frac{ u^m}{m!  } \psi_{n,m}^{(\alpha,\beta)}(z,\bz)   &
		=
		\frac{(-1)^n n! }{ z^n } e^{uz} L^{(\beta-n)}_{n}(\alpha |z|^2 -uz )  .\end{align}
\end{proposition}

\begin{proof} 	
	 The first assertion can be handled starting from the Rodrigues formula \eqref{Rodbea} and next expanding $e^{(z -v)u } $ in the second right hand-side of \eqref{genfctD}. Indeed, the identity \eqref{GA} 
	  immediately follows from Theorem \ref{HajarGen1}  by identifying the obtained series in $u$.

		For \eqref{GA5}, we have 
	\begin{align*}
		\sum_{m=0}^{+\infty}
		\frac{ u^m}{m!  } \psi_{n,m}^{(\alpha,\beta)}(z,\bz)  
		 &=  
		\frac{(-1)^n}{z^{\beta}} e^{\alpha |z|^2} \frac{\partial^n}{\partial{z^n}}(z^{\beta} e^{-\alpha|z|^2 + zu})  \nonumber \\
		%
		&=
		\frac{(-1)^n n! }{ z^n }    e^{uz} L^{(\beta-n)}_{n}  (\alpha |z|^2 -uz )  . 
	\end{align*} 
The latter expression in terms of the generalized Laguerre polynomials is immediate by making the variable change $x= \alpha |z|^2 -uz$. 
	\end{proof}



The exact statement of the next result concerning a special generating function for the $\psi_{n,m}^{\alpha,\beta}(z,\bz)$ makes appeal to the lower incomplete Gamma function defined  by 
 \begin{align}\label{HJ} \gamma (s,x) & =\int _{0}^{x}t^{s-1} e^{-t} dt , \, \Re(s)> 0.
	 \end{align}
Its expansion series reads \cite[p.337]{Magnus1966}
\begin{align}\label{IncompG1} \gamma (s,x)
	= 
  e^{-x}\sum _{k=0}^{\infty }{\frac {x^{k+s}}{(s) _{k+1}}} .
  \end{align}

\begin{theorem}\label{HajarGen2}
	 For every $\beta >0$ and $|uv|<|uz|$, we have
	\begin{align}\label{GAA}
		\sum_{m,n=0}^{+\infty}\frac{u^m v^n}{(\beta+1)_m n!}\psi_{n,m}^{\alpha,\beta}(z,\bz)= \beta u^{-\beta} z^{-\beta} e^{u(z-v)+\alpha \bz v}\gamma(\beta,u(z-v)).
	\end{align}    
\end{theorem}

\begin{proof}
The left hand-side in \eqref{GAA} can be expressed as  
	\begin{align*}
	\sum_{m,n=0}^{+\infty}\frac{u^m v^n}{(\beta+1)_m n!} \psi_{n,m}^{\alpha,\beta}(z,\bz)
	&=\beta z^{-\beta}e^{\alpha |z|^2} \sum_{n=0}^{+\infty}
	\frac{(-v)^n}{n!} \frac{\partial^{n}}{\partial{z^n}}\left(u^{-\beta}e^{-\alpha |z|^2} \left[ \sum_{m=0}^{+\infty}\frac{(zu)^{m+\beta}}{(\beta)_{m+1}}\right] \right) .
\end{align*} 
This follows making use of \eqref{Rodbea} as well as  the expansion \eqref{IncompG1}. Moreover, we get
	\begin{align*}
		\sum_{m,n=0}^{+\infty}\frac{u^m v^n}{(\beta+1)_m n!} \psi_{n,m}^{\alpha,\beta}(z,\bz)
		&= \frac{\beta}{z^\beta u^\beta} e^{\alpha |z|^2}  \sum_{n=0}^{+\infty}  \frac{(-v)^n}{n!} \frac{\partial^{n}}{\partial{z^n}}\left(e^{z(u-\alpha\bz)}\gamma(\beta,zu) \right)\\
		&= \frac{\beta}{z^\beta u^\beta} e^{zu} \sum_{n=0}^{+\infty} \sum_{k=0}^{n}   \frac{(-v)^n (-\alpha\bz )^{n-k} (-u)^k }{k!(n-k)!}  (1-\beta)_{k}\gamma(\beta-k,zu) \\
		& =  \frac{\beta}{(zu)^{\beta}}e^{zu+\alpha \bz v} \sum_{k=0}^{+\infty}\frac{u^kv^{k}}{k!}(1-\beta)_{k}\gamma(\beta-k,zu) 
		 .
	\end{align*} 
The second equality follows using the  Leibniz formula combined with the derivative formula given for the lower incomplete Gamma function in  \cite[p. 21]{Brychkov}.
	Finally, by means of the series formula in \cite[p. 460]{Brychkov} we arrive at the expression 
	$$\sum_{m,n=0}^{+\infty}\frac{u^m v^n}{(\beta+1)_m n!} \psi_{n,m}^{\alpha,\beta}(z,\bz) =  \frac{\beta}{(zu)^{\beta}} \gamma(\beta,(z-u)v) e^{zu+\alpha \bz v-uv},$$
	valid for all $z,v\in \C$ such that $|v|<|z|$.
\end{proof}

\begin{corollary}\label{particular}
	Let $u$, $v$ and $z$ be complex numbers such that $z\ne 0$, $|z|> |v|$ and $\Re(u(z-v))>0$.  Then, for every $\beta >0$  we have 
 $$ \sum_{m,n=0}^{+\infty}\frac{u^m v^n}{(\beta+1)_mn!}\psi_{n,m}^{\alpha,\beta}(z,\bz) =   \left( 1-\frac{v}{z}\right) ^{\beta} e^{\alpha \bz v} {_1F_1}\left( \begin{array}{c}1 \\ \beta+1  \end{array}\bigg |   u(z-v) \right).$$  	
\end{corollary}

\begin{proof}
	This can be handled by means of Theorem \ref{HajarGen2} and the hypergeometric representation of the  lower incomplete Gamma function \cite[p.337]{Magnus1966}  \begin{align}\label{HJJ}
		\gamma(s,x)&= \frac{e^{-x}}{s} x^s {_1F_1}\left( \begin{array}{c}1 \\  s+1  \end{array}\bigg | x \right) , \, \Re(x)> 0.
	\end{align}
		\end{proof}
	

\begin{remark}
The generating function in \eqref{GAA} can be rewritten in terms of the upper incomplete Gamma function  \cite[p.337]{Magnus1966}
\begin{align}\label{HJ3} 
	\Gamma (s,x)=\int _{x}^{\infty }t^{s-1}e^{-t}dt 
\end{align}
since $\gamma(s,x) =\Gamma(s)-\Gamma(s,x)$. Indeed, for  $\beta$ being a positive integer we have  			
	\begin{align}\label{gen}\sum_{m,n=0}^{+\infty}\psi_{n,m}^{\alpha,\beta}(z,\bz)\frac{u^m v^n}{(\beta+1)_m n!}=
	\frac{\beta \left( \Gamma(\beta) - \Gamma(\beta,u(z-v)) \right) }{(zu)^{\beta}} e^{\alpha \bz v+u(z-v)}	.
	\end{align}
\end{remark}

\begin{remark}
As immediate consequence of Theorem \ref{HajarGen2}, one can prove that 
the partial generating function  in \eqref{GA}  
remains valid for $(v, z)$ in a special region of  $\C\times \C$. 
\end{remark}
\section{Integral representations}  \label{SecInRep}

The aim below is to derive some integral representations for the considered poly-meromorphic It\^o--Hermite functions $\psi_{n,m}^{(\alpha,\beta)}$.
The first one involves the Bessel function of order $\nu>-1$ of the first kind defined by \cite[p. 675]{Brychkov}, 
$$J_\nu (z) :=\frac{1}{\Gamma(\nu +1)}\left(\frac{z}{2}\right)^{\nu} {_0F_1}\left(  \nu +1; -\frac{z^2}{4}\right) . $$ 
 More specifically, we assert the following.

\begin{proposition} For fixed real $\beta$ and integers $n,m$ such that $n=0,1,\cdots$ and $\beta +m-n >-1$, we have
	\begin{align} \label{ZZ2} \psi_{n,m}^{\alpha,\beta}(z,\bz)=(-1)^n \frac{z^{m-n} e^{\alpha|z|^2} }{(\sqrt{\alpha}|z|)^{\beta+m-n}} \int_{0}^{+\infty} x^{n+m+\beta+1} J_{\beta+m-n}(2\sqrt{\alpha} |z|x) e^{-x^2} dt. \end{align}   
\end{proposition}
\begin{proof}  Making use of the close connection of $\psi_{n,m}^{\alpha,\beta}(z,\bz)$ to the Laguerre polynomials combined with their integral representation in terms of the Bessel function  \cite[p. 243]{Magnus1966} 
	$$L_{n}^{(\mu)}(x)=\frac{x^{{-\mu}/{2}} e^x}{n!} \int_{0}^{+\infty}e^{-t}t^{n+\frac{_\mu}{2}}J_{\mu}(2\sqrt{t x})dt,$$ 
valid for $n=0,1,2, \cdots,$ and $n+\mu >-1$ with $x$ being a real positive number,   the expression \eqref{ExpLag} of $\psi_{n,m}^{\alpha,\beta}$  
	implies 
		\begin{align} \label{ZZ} \psi_{n,m}^{\alpha,\beta}(z,\bz)=(-1)^n \frac{z^{m-n} e^{\alpha|z|^2} }{(\sqrt{\alpha}|z|)^{\beta+m-n}} \int_{0}^{+\infty}e^{-t} t^{\frac{n+m+\beta}{2}}J_{\beta+m-n}(2|z|\sqrt{\alpha t})dt,
		 \end{align} 
		for every integer $m$ such that $\beta+m>-1$. Finally,  the change of variable $t=x^2$ infers the expression in \eqref{ZZ2}. 
\end{proof}

The next integral representations are are obtained by means of the   generating functions  \eqref{genfctD}  and \eqref{GA}.
  
\begin{proposition}
	The integral representation   
	\begin{align}\label{Ghanmi}
		\psi_{n,m}^{(\alpha,\beta)}(z,\bz)
		& = \frac{1}{\pi^2 z^{\beta}} \int_{\mathbb{C}^2}  u^m v^n  \left( z -\overline{v} \right)^{\beta} e^{-|u|^2- |v|^2+\alpha \overline{v}\bz +\overline{u} z-\overline{u}\overline{v}} d\lambda(u,v)
	\end{align}
	holds  for every   $\beta>-1$. Moreover, we have
		\begin{align}\label{IntREp3}
		\psi_{n,k}^{\alpha,\beta}(z,\bz)\frac{ |v|^{2j}}{j!}  = \frac{1} {z^{\beta}}
		\int_{\C} v^m(z-\overline{v})^{\beta+k} e^{-(v - \alpha  \bz) \overline{v} } d\lambda(v) . \end{align}
\end{proposition}


\begin{proof}
Thanks to  $ \psi_{m,n}^{\alpha,\beta}(z,\bz)=\overline{\psi_{m,n}^{\alpha,\beta}(\bz,z)}$, we can rewrite the generating function  \eqref{genfctD}  in the following equivalent form 
	\begin{align}\label{genfctD2}
	\sum_{m,n=0}^{+\infty}\frac{\overline{u}^m \overline{v}^n}{m!n!}\psi_{n,m}^{\alpha,\beta}(z,\bz) =\left( 1-\frac{\overline{v}}{z}\right) ^{\beta}
	e^{ \overline{u} z  +\alpha \overline{v} \bz - \overline{u} \overline{v} } . 
\end{align} 
Next, by multiplying the both sides by the monomials in $u$ and $v$ and integrating on the whole two-dimensional complex space endowed with the Gaussian measure, it follows  
\begin{align*} 
	\int_{\mathbb{C}^2} \left( 1-\frac{\overline{v}}{z}\right)^{\beta}    u^m  v^n   e^{ -|u|^2 - |v|^2 + \overline{u} z + \alpha \overline{v} \bz -\overline{u} \overline{v}}   d\lambda(u,v) 
	= 
	\pi^2 \psi_{n,m}^{\alpha,\beta}(z,\bz),
\end{align*}
which leads to \eqref{Ghanmi}.
Analogously,  one gets \eqref{IntREp3} starting from \eqref{GA}. 
\end{proof}

The generating function in Theorem \ref{HajarGen1} can be also  employed to establish the following integral representation.

\begin{proposition} Let $\beta$ be an integer such that $\beta+m \geq 0$. Then, we have  
		\begin{align}\label{GA3}	\psi_{n,m}^{\alpha,\beta}(z,\bz) 
		&	=   \frac{(-1)^{m+\beta} \alpha^{n+1}}{ \pi z^{\beta} } 	\int_{\C}   \xi^n 	\overline{\xi}^{m+\beta} e^{-\alpha (|\xi|^2 - |z|^2 + \xi z - \overline{\xi} \bz ) }	d\lambda(\xi) .\end{align} 
\end{proposition}

\begin{proof} 
Note that  making use of the $2d$ fractional Fourier transform $(1.2)$ introduced in \cite{BenahmadiGh2018} one obtains the following integral formula
\begin{align}\label{expo}
	e^{\alpha zw} & =   \frac{\alpha}{\pi} 
	\int_{\C}  e^{-\alpha |\xi|^2 + \alpha (\xi z + \overline{\xi} w ) } d\lambda(\xi)
\end{align}  
for every complex numbers $z,w$ and real $\alpha > 0 $. It can also be viewed as a reproducing property for the reproducing kernel of the Segal--Bargmann space.
Next, by rewriting the generating function in  \eqref{GA} in the following equivalent form
		\begin{align}\label{GA2}
	\sum_{n=0}^{+\infty}\psi_{n,k-\beta}^{\alpha,\beta}(z,\bz)\frac{ v^n}{n!}
&	= \frac{(-1)^k}{\alpha^k z^{\beta} } e^{\alpha|z|^2}  \frac{\partial^k }{\partial \bz^k } \left( e^{\alpha ( v-z)\bz}  \right)
 \end{align}
and making appeal to the formula \eqref{expo}, it follows 
	\begin{align*} 	\sum_{n=0}^{+\infty}\psi_{n,k-\beta}^{\alpha,\beta}(z,\bz)\frac{ v^n}{n!}
	&	= \frac{(-1)^k}{\alpha^k z^{\beta} } \frac{\alpha}{\pi}     
	\int_{\C} 
	 (\alpha\overline{\xi})^k e^{-\alpha (|\xi|^2 - |z|^2)+ \alpha (\xi(v-z) +\overline{\xi} \bz ) }
	d\lambda(\xi)  \\
	&	= \sum_{n=0}^\infty \frac{v^n}{n!}   \left(  \frac{(-1)^k}{  z^{\beta} } \frac{\alpha}{\pi} \alpha^n 	\int_{\C}   \xi^n 	\overline{\xi}^k e^{-\alpha (|\xi|^2 - |z|^2 + \xi z - \overline{\xi} \bz ) }	d\lambda(\xi) \right) .
\end{align*}
The result in \eqref{GA3} is then immediate by identification.
\end{proof}

%

\begin{remark}
	The identity \eqref{GA3} can be reproved starting from \eqref{RelHermite} and making use of the classical integral representation of the It\^o--Hermite polynomials in \cite{Gh2017Mehler}.
\end{remark}


\section{Applications}  \label{SubsIntTrans}

 \subsection{New integral formula for the generalized Laguerre polynomials}
Using the obtained results one can derive new interesting integral formulas for the generalized Laguerre polynomials.
Thus, we claim the following.
\begin{theorem} The
	integral identity  
	\begin{align} \label{INtLaguerre}
		L^{(\beta-n)}_{n}  (\alpha |z|^2 -uz )  = \frac{	(-1)^n z^{n-\beta}}{ n! \pi }  	\int_{\C} \overline{v}^n \left( z -v\right) ^{\beta}
		e^{ \alpha v  \bz - u v }  e^{- |v|^2} d\lambda(v) 
	\end{align} 
	holds  for  $\beta>-1$.
\end{theorem}

\begin{proof}
	From  \eqref{genfctD}, one has
	\begin{align*} 
		\sum_{m=0}^{+\infty}\frac{u^m }{m! }\psi_{k,m}^{\alpha,\beta}(z,\bz) 
		&=
		\frac{1}{\pi}  \scal{\sum_{m,n=0}^{+\infty}\frac{u^m v^n}{m! n!}\psi_{n,m}^{\alpha,\beta}(z,\bz)  , v^k}_{L^2(\C,e^{-|v|^2})}  
		\\&=
		\frac{e^{uz}}{\pi z^\beta} \int_{\C} \overline{v}^k \left( z -v\right) ^{\beta}
		e^{ \alpha v  \bz - u v }  e^{- |v|^2} d\lambda(v) .  
	\end{align*} 
	Accordingly, the proof of \eqref{INtLaguerre} readily follows from \eqref{GA5}.
\end{proof}

\begin{remark}
	As particular case we get 
	\begin{align} 
		L^{(\beta-n)}_{n}  ( z )  = \frac{	(-1)^n z^{n-\beta}}{ n! \pi }  	\int_{\C} \overline{v}^n \left( z -v\right) ^{\beta}
		e^{- v(\overline{v} -1 ) } d\lambda(v)   
	\end{align} 
	for any $z\in \C$ by specifying $\alpha=0$ and $ u=-1 $. Also, by taking $\alpha =1 $ and $u=0$, we get 
	\begin{align} 
		L^{(\beta-n)}_{n}  (  |z|^2   )  = \frac{	(-1)^n z^{n-\beta}}{ n! \pi }  	\int_{\C} \overline{v}^n \left( z -v\right) ^{\beta}
		e^{   v  \bz  }  e^{- |v|^2} d\lambda(v) .
	\end{align} 
\end{remark}

 \subsection{Associated integral transforms}

The orthogonality property of the functions $\psi_{n,m}^{\alpha,\beta}$ suggest the consideration of two special functional spaces $ \mathcal{F}^{2,\alpha}_{\beta,n}(\C)$ and $ \widetilde{\mathcal{F}}^{2,\alpha}_{\beta,m}(\C)$ of  poly-meromorphic or anti-poly-meromorphic functions on the punctured complex plane in
$L^{2,\alpha}_{\beta}(\mathbb{C}) $ for fixed non-negative integers $m$ and $n$ with $m+\beta >-1$. These spaces are spanned by  $\psi_{n,j}^{\alpha,\beta}$, $j\geq  [ -\beta ]$, and  $\psi_{k,m}^{\alpha,\beta}$, $k= 0,1,2, \cdots$, respectively.  
Moreover, they can be seen as the  poly-meromorphic analogs of the true  poly-analytic (and anti-poly-analytic) Bargmann spaces \cite{AbreuFeichtinger2014,Vasilevski2000} defined as specific closed subspace in 
$ \ker\left( {\partial^{n+1}}/{\partial \bz^{n+1}} \right)  \cap L^{2,\alpha}_g(\C) ,$
and realized also as $L^2$-eigenspace 
$\mathcal{F}^{2,\alpha}_n(\C) = \ker( \Delta_\alpha  - \alpha n)$ associated with the $n$-th Landau levels of  the self-adjoint magnetic Laplacian 
$$
\Delta_\alpha  =  \Delta_{\alpha,0}   =  -  \frac{\partial^2}{\partial z\partial\bz  } + \alpha \bz  \frac{\partial}{\partial \bz }
$$
acting on $L^{2,\alpha}_g(\C)$ (see \cite{AskourIntissarMouayn2000,Shigekawa1987}). 

Next, we show that $ \widetilde{\mathcal{F}}^{2,\alpha}_{\beta,m}(\C) = \overline{Span\{ \psi_{k,m}^{\alpha,\beta}, \, k=0,1,2, \cdots \} }^{L^{2,\alpha}_{\beta}(\mathbb{C})}$  can be realized as the image of the classical Segal--Bargmann space $ \mathcal{F}^{2,\alpha}(\C)=Hol(\C) \cap L^{2,\alpha}_g(\C)$ of holomorphic functions in the Hilbert space of the Gaussian functions, $L^{2,\alpha}_g(\C) :=L^2(\C, e^{-\alpha|\xi|^2} d\lambda)$, by means of the specific integral transform 
\begin{align}\label{Transform2GA}
	\mathcal{B}_{m}^{\alpha,\beta}f (z) &:=
	\frac{\alpha}{\pi} \left( 	\frac{\alpha^{\beta+m}}{ \Gamma(\beta+m+1)}\right)^{1/2}  z^{m}
	\int_{\C}   \left( 1-\frac{\bw}{z} \right) ^{\beta +m}  e^{-\alpha \bw (w- \bz)}  f(w) d\lambda(w)  ,
\end{align}
provided that the integral exists. Here 
  $m$ is a fixed integer  such that $m> -\beta-1$ for given   $\beta > -1$.  
	In fact, the transform $\mathcal{B}_{m}^{\alpha,\beta}$ is well defined and maps the orthonormal basis $e_n^\alpha(z) = (\alpha^{n+1}/\pi n!)^{1/2}z^n$ of $ \mathcal{F}^{2,\alpha}(\C)$ to an orthonormal basis of $\widetilde{\mathcal{F}}^{2,\alpha}_{\beta,m}(\C) $. More precisely, we have 
		$$ 	\mathcal{B}_{m}^{\alpha,\beta} (e_n) (z)  = \left(  \frac{\alpha^{\beta+m+1}}{\pi \alpha^n \Gamma(\beta+m+1)  n! }\right)^{1/2} \psi_{n,m}^{\alpha,\beta}(z,\bz).$$ 
	This follows by observing that the integral kernel of the transform $	\mathcal{B}_{m}^{\alpha,\beta}$ in \eqref{Transform2GA} is the generating function in \eqref{GA}.
			  Its inverse is given by 
	\begin{align}
		(\mathcal{B}_{m}^{\alpha,\beta})^{-1} (f)(w)&= 	\frac{\alpha}{\pi} \left( 	\frac{\alpha^{\beta+m}}{ \Gamma(\beta+m+1)}\right)^{1/2}  
		\int_{\C}  \frac{ (\bz-w)^{\beta +m} }{\bz^{\beta}} e^{-\alpha( \bz - w) z } f(z) d\lambda(z).
	\end{align}
It is worth noticing that for the particular case of $\beta=0$, the corresponding transform $\mathcal{B}_{m}^{\alpha,0}$ reduces further to the one considered in \cite[Remark 2.13]{BenahmadiGh2018} mapping unitarily the Segal-Bargmann space to the true  anti-poly-analytic Bargmann spaces  $\widetilde{\mathcal{F}}^{2,\alpha}_m(\C) $. 

Similarly, associated with the kernel function on $\C\times \C$ given through the partial generating function in \eqref{GA5},  
\begin{align}\label{kernelfct1}
	s_{n}^{(\alpha,\beta)}(u,z)   &
	:	=
	\frac{(-1)^n n! }{ z^n }    e^{uz} L^{(\beta-n)}_{n}  (\alpha |z|^2 -uz ) ,
\end{align}
we consider the integral transform
\begin{align}\label{Transform1GA5}
	\mathcal{S}_{n}^{\alpha,\beta}f (z)  &:= 
	\frac{(-1)^n n! }{ z^n }   \int_{\C}  L^{(\beta-n)}_{n}  (\alpha |z|^2 -uz )   e^{-u(\overline{u} - z)} f(u)  d\lambda(u).
\end{align}
The image of 
	$ \mathcal{F}^{2,\alpha}(\C)$ by $	\mathcal{S}_{n}^{\alpha,\beta}$ for arbitrary $\beta >-1$ is the closed subspace of $L^{2,\alpha}_{\beta}(\mathbb{C}) $ spanned by  $\psi_{n,j}^{\alpha,\beta}$ for varying $j=0,1,2, \cdots$,
$$	\mathcal{S}_{n}^{\alpha,\beta}(\mathcal{F}^{2,\alpha}(\C)) =  \overline{Span\{ \psi_{n,j}^{\alpha,\beta}; \, j=0,1,2, \cdots \} }^{L^{2,\alpha}_{\beta}(\mathbb{C})} =: \widehat{\mathcal{F}}^{2,\alpha}_{\beta,n}(\C) .$$
For  $\beta >0$, it reduces to the restricted $\beta$-poly-meromorphic space $\widehat{\mathcal{F}}^{2,\alpha}_{\beta,n}(\C) $  which is strictly contained in $\mathcal{F}^{2,\alpha}_{\beta,n}(\C) : = \overline{Span\{ \psi_{n,j}^{\alpha,\beta}; \, j= [-\beta], [-\beta]+1,\cdots \} }^{L^{2,\alpha}_{\beta}(\mathbb{C})}$.
For $-1 < \beta \leq 0$, this is exactly the poly-meromorphic space $ \mathcal{S}_{n}^{\alpha,\beta}(\mathcal{F}^{2,\alpha}(\C)) =  \widehat{\mathcal{F}}^{2,\alpha}_{\beta,n}(\C) = \mathcal{F}^{2,\alpha}_{\beta,n}(\C)$

\end{document}